\documentclass{article}
\usepackage[utf8]{inputenc}
\usepackage[left=2.5cm,right=2cm,top=2cm,bottom=2.5cm]{geometry}

\usepackage{authblk}

\usepackage{graphicx}
\usepackage{tikz}
\usetikzlibrary{positioning,chains,fit,shapes,calc}
\usepackage[numbers]{natbib}
\usepackage{bookmark}
\usepackage{graphicx}%
\usepackage{multirow}%
\usepackage{amsmath,amssymb,amsfonts}%
\usepackage{amsthm}%
\usepackage{mathrsfs}%
\usepackage[title]{appendix}%
\usepackage{graphicx}
\usepackage{hyperref}
\usepackage{algorithm}
\usepackage{algpseudocode}
\usepackage{xcolor}%
\usepackage{textcomp}%
\usepackage{manyfoot}%
\usepackage{booktabs}%
\usepackage{listings}%
\usepackage{float}
\algtext*{EndWhile}
\algtext*{EndIf}%
\algtext*{EndFor}

\newtheorem{theorem}{Theorem}
\newtheorem{lemma}{Lemma}%
\newtheorem{corollary}{Corollary}%
\usetikzlibrary{positioning,chains,fit,shapes,calc}

\renewcommand{\|}{|}

\newcommand{\cala}{\mathcal{A}}
\newcommand{\calt}{\mathcal{T}}

\newcommand{\OPT}{\ensuremath{\mathit{\textit{Opt}}}}

\newcommand{\MF}{\ensuremath{X}}
\newcommand{\WW}{\ensuremath{W}}
\newcommand{\PW}{\ensuremath{U_W}}
\newcommand{\PWL}{\ensuremath{U_W'}}
\newcommand{\PWLL}{\ensuremath{U_W''}}
\newcommand{\SUP}{\ensuremath{\mathit{Sup}}}

\newcommand{\const}{\ensuremath{2^{2^K}2^K}}

\newcommand{\Vmin}{\ensuremath{v_{\min}}}

\DeclareMathOperator*{\argmin}{arg\,min}

\begin{document}

\definecolor{myblue}{RGB}{80,80,160}
\definecolor{mygreen}{RGB}{80,160,80}

\definecolor{mygreen2}{RGB}{213,232,212}
\definecolor{myorange}{RGB}{255,230,204}
\definecolor{myblue2}{RGB}{218,232,252}
\definecolor{myred}{RGB}{248,206,204}
\definecolor{mygrey}{RGB}{245,245,245}

\title{Approximation algorithms for\\ the MAXSPACE advertisement problem}
\author[1]{Mauro R. C. da Silva}
\author[1]{Lehilton L. C. Pedrosa}
\author[1]{Rafael C. S. Schouery}
\affil[1]{Institute of Computing\\ University of Campinas}

\affil[ ]{\textit {maurorcsc@gmail.com, \{lehilton, rafael\}@ic.unicamp.br}}

\maketitle

%
%
%
%

\section*{Abstract}
In MAXSPACE, given a set of ads $\cala$, one wants to schedule a subset ${\cala'\subseteq\cala}$ into $K$ slots ${B_1, \dots, B_K}$ of size $L$.
Each ad~${A_i \in \cala}$ has a \textit{size}~$s_i$ and a \textit{frequency}~$w_i$.
A schedule is feasible if the total size of ads in any slot is at most $L$, and
each ad ${A_i \in \cala'}$ appears in exactly $w_i$ slots and at most once per slot.
The goal is to find a feasible schedule that maximizes the sum of the space occupied by all slots.
We consider a generalization called MAXSPACE-R for which an ad~$A_i$ also has a release date~$r_i$ and may only appear in a slot~$B_j$ if ${j \ge r_i}$.
For this variant, we give a~\mbox{$1/9$-approximation} algorithm.
Furthermore, we consider MAXSPACE-RDV for which an ad~$A_i$ also has a deadline~$d_i$ (and may only appear in a slot~$B_j$ with $r_i \le j \le d_i$), and a value~$v_i$ that is the gain of each assigned copy of~$A_i$ (which can be unrelated to~$s_i$).
We present a polynomial-time approximation scheme for this problem when $K$ is bounded by a constant. This is the best factor one can expect since MAXSPACE is strongly NP-hard, even if $K = 2$.
%


\textbf{keywords:} Approximation Algorithm, PTAS, Scheduling of Advertisements, MAXSPACE.

%


\section{Introduction}
\label{intro}

Many websites (such as Google, Yahoo!, Facebook, and others) offer free services while displaying advertisements (or ads) to users. Each website often has a single strip of fixed height, which is reserved for scheduling ads, and the set of displayed ads changes on a time basis. For such websites, advertisement is the primary source of revenue. Thus, it is essential to find the best way to dispose the ads in the available time and space while maximizing the revenue~\citep{kumar2015optimization}.

The revenue from web advertising grew considerably in the 21st century. In~2022, the total revenue was~US\$209{.}7 billion, an increase of~10{.}8\% from the previous year. It is estimated that web advertising comprised~52\% of all advertising spending, overtaking television advertising. In 2022, banners and search engine ads comprised~70.5\% of internet advertising, representing a revenue of~US\$147{.}9 billion~\citep{2022iab}. Web advertising has created a multi-billionaire industry where algorithms for scheduling advertisements play an important role.

Websites like Facebook and Mercado Livre (a large Latin American marketplace) use banners to display advertisements while users browse. Google displays ads sold through Google Ad Words in its search results within a limited area, in which ads are in text format and have sizes that vary according to the price (see Figure~\ref{fig:google_ads}).

\begin{figure}[H]
    \centering
    \includegraphics[width=1\textwidth]{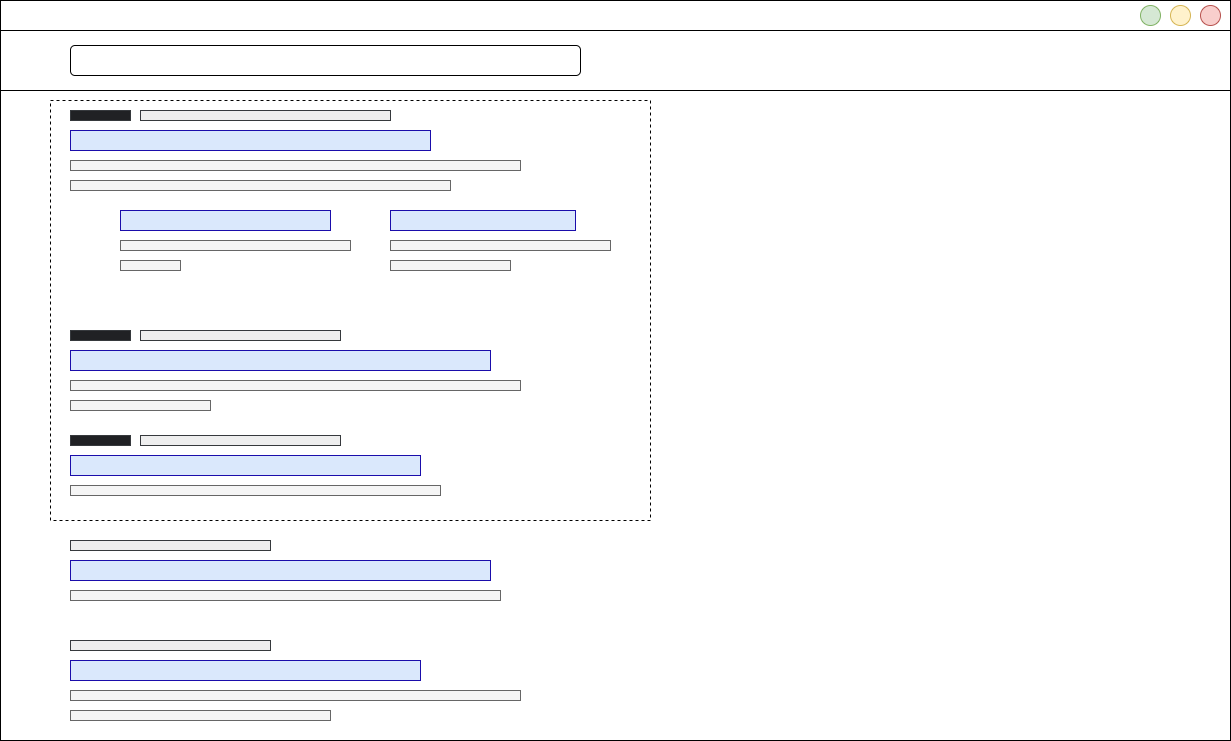}
    \caption{Example of search engine homepage with ads with variable sizes displayed as results in a limited space, represented by the dotted area.\label{fig:google_ads}}
\end{figure}


We consider the class of Scheduling of Advertisements problems introduced by~\citet{adler2002scheduling}, where, given a set~${\cala = \{A_1, A_2, \dots, A_n\}}$ of advertisements, the goal is to schedule a subset~${\cala' \subseteq \cala}$ into a banner in~$K$ equal time-intervals. The set of ads scheduled to a particular time interval~$j$,~${1 \le j \le K}$, is represented by a set of ads~${B_j \subseteq \cala'}$, which is called a \textit{slot}. Each ad~$A_i$ has a \textit{size}~$s_i$ and a \textit{frequency}~$w_i$ associated with it. The size~$s_i$ represents the amount of space~$A_i$ occupies in a slot, and the frequency~${w_i \leq K}$ represents the number of slots that should contain a copy of~$A_i$. An ad~$A_i$ can be displayed at most once in a slot, and~$A_i$ is said to be \textit{scheduled} if~$w_i$ copies of~$A_i$ appear in slots with at most one copy per slot~\citep{adler2002scheduling, dawande2003performance}.

The main problems in this class are MINSPACE and MAXSPACE\@. In MINSPACE, all ads have to be scheduled, and the goal is to minimize the fullness of the fullest slot. In MAXSPACE, an upper bound~$L$ is specified, representing each slot's size. A feasible solution for this problem is a schedule of a subset~${\cala' \subseteq \cala}$ into slots~${B_1, B_2, \dots, B_K}$, such that each~${A_i \in \cala'}$ is scheduled and the fullness of any slot does not exceed the upper bound~$L$, that is, for each slot~$B_j$,~${\sum_{A_i \in B_j}{s_i}\leq L}$. MAXSPACE aims to maximize the slots' fullness, defined by~$\sum_{A_i \in \cala'}{s_i w_i}$. Both problems are strongly NP-hard~\citep{adler2002scheduling, dawande2003performance}.

Even though these problems where introduced with advertisement in mind, since MAXSPACE and MINSPACE are packing problems, they can be applied to pack several kinds of items into bins or slots. For example, a solution for this problem can populate the columns of the social photo-sharing network Pinterest~\citep{pinterest} and other sites with the same kind of layout (called \emph{grid~layout}).

\subsection{Previous Works}

In the literature, there are works regarding approximation and exact algorithms for MINSPACE and MAXSPACE. Also, some special cases of these problems where defined by~\citet{dawande2003performance}. In~MAX$_w$, every ad has the same frequency $w$. In~MAX$_{K|w}$, every ad has the same frequency~$w$, and the number of slots~$K$ is a multiple of~$w$. Moreover, in MAX$_s$, every ad has the same size $s$. Analogously, they define three special cases of MINSPACE:~MIN$_w$,~MIN$_{K|w}$~and~MIN$_s$.
%

Regarding approximation algorithms for MAXSPACE, \citet{adler2002scheduling} present a~$\frac{1}{2}$-approximation when the ad sizes form a sequence~${s_1 > s_2 > \cdots > s_n}$, such that for all~$i$,~$s_i$ is a multiple of~$s_{i+1}$. \citet{dawande2003performance} present three approximation algorithms:
a~${(\frac{1}{4} + \frac{1}{4K})}$-approximation for MAXSPACE, a~$\frac{1}{3}$-approximation for~MAX$_w$ and a~$\frac{1}{2}$-approximation for~MAX$_{K|w}$. \citet{freund2002approximating} proposed a~${(\frac{1}{3} - \varepsilon)}$-approximation for MAXSPACE and a~${(\frac{1}{2} - \varepsilon)}$-approximation for the special case in which the size of the ads are in the interval~${[L/2, L]}$.

For MINSPACE, \citet{adler2002scheduling} present a~$2$-approximation called~\emph{Largest-Size Least-Full} (LSLF) which is also a ${(\frac{4}{3}-\frac{w}{3K})}$-approximation to MIN$_{K|w}$~\citep{dawande2003performance}. \citet{dawande2003performance} present a~$2$-approximation for MINSPACE using \emph{LP Rounding}, and \citet{dean2003improved} present a~$\frac{4}{3}$-approximation for MINSPACE using Graham's algorithm for schedule~\citep{graham1979optimization}.

From the exact-algorithm standpoint, \citet{kaul2018optimal} present an integer programming model for placing advertisements optimally in a two-dimensional banner. \citet{kim2020online} present a variant of MAXSPACE with a new objective function that includes factors that influence advertising effectiveness in terms of click-through rate. They provide an integer programming model and two meta-heuristics for this problem.

\subsection{Our results}

In practice, the time interval relative to each slot in scheduling advertising can represent minutes, seconds, or long periods, such as days and weeks. One often considers the idea of \textit{release dates} and \textit{deadlines}. An ad's release date indicates the beginning of its advertising campaign. Analogously, the deadline of an ad indicates the end of its advertising campaign. For example, ads for Christmas must be scheduled before December~25th. 

With this in mind, we consider a MAXSPACE generalization called MAXSPACE-R in which each ad~$A_i$ has one additional parameter, a release date~${r_i \geq 1}$. The release date of ad~$A_i$ represents the first slot where a copy of~$A_i$ can be scheduled; that is, a copy of $A_i$ cannot be scheduled in a slot~$B_j$ with~${j < r_i}$. In MAXSPACE-R, we assume that the frequency of each ad~$A_i$ is compatible with its release date, that is,~${K - r_i + 1 \geq w_i}$. 

Notice that in the original MAXSPACE and in MAXSPACE-R, the value of an ad corresponds to the space it occupies multiplied by the number of times it appears. In practice, the value of an ad can be influenced by other factors, such as the expected number of clicks it generates for the advertiser~\citep{briggs1997advertising}. The number of times the ad appears can also be influenced by other factors, such as the budget provided by the advertiser.

In order to consider that the value of the ad is not necessarily related to its size and to consider deadlines, we also consider a MAXSPACE-R generalization called MAXSPACE-RDV in which each ad~$A_i$ has a deadline~${d_i \leq K}$ and a value~$v_i$. Similarly to the release date, the deadline of an ad~$A_i$ represents the last slot where we can schedule a copy of~$A_i$; thus~$a_i$ cannot be scheduled in a slot~$B_j$ with~${j > d_i}$. We assume that the frequency of each ad~$A_i$ is compatible with its release date and deadline, that is,~${w_i \leq d_i - r_i + 1}$. In MAXSPACE-RDV, each assigned copy of an ad~$A_i$ also has a value~$v_i$, and the value of a solution is the sum of~$v_i w_i$ for each scheduled ad~$A_i$. Note that~$v_i$ can be unrelated to the size~$s_i$ of~$A_i$.


Let~$\Pi$ be a maximization problem. A family of algorithms~$\{H_\varepsilon\}$ is a \textit{Polynomial-Time Approximation Scheme}~(PTAS) for~$\Pi$ if, for every constant~${\varepsilon > 0}$,~$H_\varepsilon$ is a~$(1-\varepsilon)$-approximation for~$\Pi$~\citep{vazirani2013approximation}.
A~\textit{Fully Polynomial-Time Approximation Scheme}~(FPTAS) is a PTAS whose running time is also polynomial in~$1/\varepsilon$. Notice that MAXSPACE does not admit an FPTAS even for~${K = 2}$, since it generalizes the \textit{Multiple Subset Sum Problem} with identical capacities~(MSSP-I), which does not admit an FPTAS even for~${K = 2}$~\citep{kellerer2004introduction}.

In a previously published conference paper~\citep{da2019polynomial}, we proposed a PTAS for MAXSPACE-RD with bounded~$K$, which is the particular case of MAXSPACE-RDV where $v_i = s_i$ for every ad~$A_i$. We improve this result, using a different technique, by presenting a PTAS for MAXSPACE-RDV\@. We also present a~$1/9$-approximation algorithm for MAXSPACE-R (where $K$ is not necessarily bounded by a constant).

In the $1/9$-approximation to MAXSPACE-R, we divide the ads into large, medium, and small. We create an optimal and polynomial dynamic programming algorithm for large ads based on the classic dynamic programming for the Knapsack Problem~\citep{kellerer2004introduction}, and we use algorithms based on the best-fit heuristic to allocate medium and small ads. We also execute a step based on a local search for small ads to relocate them between slots when possible. As our problem has release dates, we cannot use the area bounds as in previous works for MAXSPACE since an optimal solution does not necessarily have a good slot fullness. Thus, it is necessary to compare the algorithm solution with the ads' allocation in an optimal solution to show the approximation factors.

In the PTAS for MAXSPACE-RDV, since the number of slots is bounded by a constant, we enumerate the most valuable ads on the solution and all possible solutions involving these ads in polynomial time. To schedule the other ads, we use a linear program algorithm to obtain a relaxed allocation of ads to slots. We give an algorithm that rounds off the fractional assignment, showing that the losses are small. This algorithm is inspired by the approximation scheme presented by \citet{frieze1984approximation} for the $m$-dimensional knapsack problem. However, in MAXSPACE, the ads have copies; thus, assigning and rounding the less valuable advertisements requires different and more elaborate techniques.
We hope these strategies can be adapted to similar packing problems, specifically with release dates and/or deadlines.

In Section~\ref{sec:1/9approx}, we present a~$1/9$-approximation algorithm for MAXSPACE-R, and in Section~\ref{sec:ptasrd}, we present a PTAS for MAXSPACE-RDV with a constant number of slots. In Section~\ref{sec:con}, we discuss the results and future works.

\section{A 1/9-approximation for MAXSPACE-R\label{sec:1/9approx}}

This section presents a~$1/9$-approximation algorithm for \mbox{MAXSPACE-R}. For this problem, we assume that~$K$ is polynomially bounded, as otherwise, the size of the solution is not polynomially bounded. We also assume that~${L = 1}$ and~${0 < s_i \le 1}$ for each~$A_i \in \cala$. We partition the ads into three sets: 
the set~${G = \{A_i \in \cala \mid s_i > 1/2\}}$ of large ads, 
the set~${M = \{A_i \in \cala \mid 1/4 < s_i \leq 1/2\}}$ of medium ads, and
the set~${P = \{A_i \in \cala \mid s_i \leq 1/4\}}$ of small ads.

Let~$S$ denote a feasible solution~${\cala'\subseteq \cala}$ scheduled into slots~${B_1, B_2, \dots, B_K}$.
Then the \textit{fullness} of a slot~$B_j$ is defined as~${f(B_j) = \sum_{A_i \in B_j} s_i}$. Also, the fullness of solution~$S$ is~${f(S) = \sum_{j = 1}^{K}{f(B_j)}}$.

In Section~\ref{sec:pg}, we present an exact algorithm for large ads; in Section~\ref{sec:pm}, we present a~$1/4$-approximation for medium ads; and, in Section~\ref{sec:pp}, we present a~$1/4$-approximation for small ads. Moreover, in Section~\ref{sec:pcasogeral}, we combine these algorithms to obtain a~$1/9$-approximation for the whole set of ads~$\cala$.

\subsection{An exact algorithm for large ads\label{sec:pg}}

We present an exact polynomial-time algorithm for large ads based on the dynamic programming algorithm for the Binary Knapsack Problem~\citep{kellerer2004introduction}.

An instance of the Binary Knapsack Problem consists of a container with capacity~$W$ and a set~$I$ of items. Each item~${i \in I}$ has a profit~$v_i$ and a weight~$p_i$. The goal is to find a subset~${I' \subseteq I}$ that maximizes the total profit and such that the sum of the weights does not exceed the container's capacity, that is,~${\sum_{i \in I'}{p_i} \leq W}$.

We say that an ad~$A_i$ appears \textit{sequentially} in a schedule~$S$ if, for each pair of slots~$B_j$ and~$B_k$ that have copies of~$A_i$ in~$S$, there is a copy of~$A_i$ in each slot~$B_\ell$ of~$S$, with~${j \le \ell \le k}$. Notice that an ad that has only one copy is always sequentially scheduled.


\begin{lemma}\label{lema:pd}
    Let~$S$ be a feasible schedule with ads of~$G$ in~$K$ slots and let~${\cala' \subseteq G}$ be the set of ads scheduled in~$S$. There is a feasible schedule~$S'$ in which all ads of~$\cala'$ appear sequentially and in non-decreasing order of release dates, that is, if some ad~$A_i$ appears before an ad~$A_j$ then~${r_i \le r_j}$.
\end{lemma}

\begin{proof}
    Note that it is not possible to add more than one copy of any ad per slot since~${s_i > 1/2}$ for all~${A_i \in G}$. Let~$S'$ be a schedule of~$\cala'$ in which the number of ads that appear sequentially is maximum. Assume by contradiction that there is an ad~$A_i$ in~$S'$ that does not appear sequentially. Let~$X$ be the set of ads with at least one copy between the first and last copies of~$A_i$ in~$S'$. Figure~\ref{fig:slinha} shows the schedule of~$S'$.

    \begin{figure}[ht]
    \centering
  \resizebox{.9\textwidth}{!}{%
  \begin{tikzpicture}
  
  \begin{scope}[start chain=going,node distance=10mm]
  \draw[black, fill=mygreen2] (0, 0) rectangle ++(3, 1);
  \draw[black, fill=myorange] (3, 0) rectangle ++(5, 1);
  \draw[black, fill=myblue2] (8, 0) rectangle ++(3, 1) node[pos=.5] {\LARGE$A_i$};
  \draw[black, fill=myred] (11, 0) rectangle ++(4, 1) node[pos=.5] {\LARGE$X$};
  \draw[black, fill=myblue2] (15, 0) rectangle ++(3, 1) node[pos=.5] {\LARGE$A_i$};
  \draw[black, fill=myred] (18, 0) rectangle ++(2, 1) node[pos=.5] {\LARGE$X$};
  \draw[black, fill=myblue2] (20, 0) rectangle ++(1, 1) node[pos=.5] {\LARGE$A_i$};
  \draw[black, fill=mygrey] (21, 0) rectangle ++(3, 1);
  \end{scope}
 
  \end{tikzpicture}
  }
  \caption[Schedule of ads of~$S'$.]{Schedule of ads of~$S'$. In blue is the ad~$A_i$ and in red the ads of set~$X$. The rest of colors represent the other ads in this schedule.\label{fig:slinha}}
\end{figure}
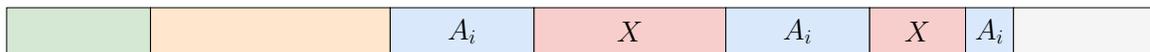
\nopagebreak


    Exchange the last copies of~$A_i$ with copies of ads in~$X$ such that~$A_i$ appears sequentially in the solution maintaining the order of X. Since the scheduling of copies of ads in~$X$ is delayed, their release dates are respected. Moreover, each copy of $A_i$ is moved to some slot after the first copy of $A_i$. Thus, the release date of $A_i$ is also respected. Therefore, the modified schedule is feasible, as seen in Figure~\ref{fig:slinhanovo}.
    However, the number of ads sequentially scheduled increases, which contradicts the assertion that~$S'$ has a maximum number of ads sequentially scheduled.

    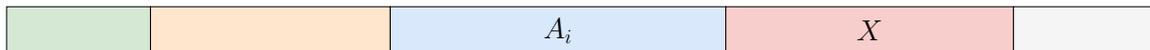
\begin{figure}[ht]
    \centering
  \resizebox{.9\textwidth}{!}{%
  \begin{tikzpicture}
  
  \begin{scope}[start chain=going,node distance=10mm]
  \draw[black, fill=mygreen2] (0, 0) rectangle ++(3, 1);
  \draw[black, fill=myorange] (3, 0) rectangle ++(5, 1);
  \draw[black, fill=myblue2] (8, 0) rectangle ++(7, 1) node[pos=.5] {\LARGE$A_i$};
  \draw[black, fill=myred] (15, 0) rectangle ++(6, 1) node[pos=.5] {\LARGE$X$};
  \draw[black, fill=mygrey] (21, 0) rectangle ++(3, 1);
  \end{scope}
 
  \end{tikzpicture}
  }
  \caption[Schedule of ads of~$S'$ after move~$A_i$.]{Schedule of ads of~$S'$ after move copies of~$A_i$ to appears sequentially.\label{fig:slinhanovo}}
\end{figure}
\nopagebreak




    Now, assume there are copies of ads in~$S'$, which are not ordered by release dates. Then, let~$A_i$ be an ad whose copies are scheduled immediately after the copies of an ad $A_j$ with $r_j > r_i$.
    Let~$B_k$ be the slot in which the first copy of~$A_j$ appears in~$S'$, then $r_j \le k$, and thus $r_i \le k$.
    We exchange the order of the ads $A_i$ and $A_j$, that is, we move each copy of~$A_i$ scheduled in a slot $B_z$ to a slot~$B_{z - w_j}$, and each copy of~$A_j$ scheduled in a slot $B_z$ to a slot~$B_{z + w_i}$.
    Since the first copy of $A_i$ is scheduled in slot~$B_k$, and copies of $A_j$ are only delayed, the modified schedule is feasible.
    By repeating this process, we obtain a schedule in which all ads appear in non-decreasing order of release dates.
\end{proof}

In Lemma~\ref{lema:pd}, we show that, given a feasible schedule~$S$, it is possible to construct a feasible schedule~$S'$ with the same set of ads in~$S$, but so that all ads appear sequentially in slots and ordered by release dates. Thus, given an instance for MAXSPACE-R with large ads~$G$, we built an instance of the Knapsack Problem, in which each ad~$A_i$ of~$G$ is an item~$j$ with weight~${p_j = w_i}$ and profit~${v_j = s_i w_i}$. Since it is not possible to add more than one ad of~$G$ per slot, we can ignore the height dimension of slots, given by~$L$, and use only the width dimension, given by the number of slots~$K$. We define the knapsack capacity as~${W = K}$. In Algorithm~\ref{pd}, we present a dynamic programming algorithm for the binary knapsack problem with release date restrictions. In this algorithm, $m[i,j]$ corresponds to the optimal value of scheduling a set of~$i$ ads with the smallest release dates to the first~$j$ slots. 


\begin{algorithm}[H]\scriptsize
\caption{Dynamic programming algorithm for large ads.\label{pd}}
\begin{algorithmic}[1]

\Procedure{DP}{$G$}
    \State creates a matrix $m[0 \dots \|G\|][0 \dots K]$
    \For{$j \gets 1, \dots, K$}
        \State $m[0, j] \gets 0$
    \EndFor
    \For{$i \gets 1, \dots, \|G\|$}
        \State $m[i, 0] \gets 0$
    \EndFor
    \For{\textbf{each }$A_i \in G$ in non-decreasing order of~$r_i$}
        \State $x \gets r_i + w_i - 1$
        \For{$j \gets 1, 2, \dots, x-1$}
            \State $m[i, j] \gets m[i-1, j]$
        \EndFor
        \For{$j \gets x, x+1, \dots, K$}
            \State $m[i, j] \gets \max\{m[i-1, j], m[i-1, j-w_i] + s_i w_i\}$
        \EndFor
    \EndFor
    \State backtrack in matrix $m$ and return the solution
\EndProcedure

\end{algorithmic}
\end{algorithm}



\begin{lemma}\label{otimo:pd}
    Algorithm~\ref{pd} runs in polynomial time in the instance size, returns a feasible solution, and is optimal for~$G$.
\end{lemma}

\begin{proof}
    The time complexity of the algorithm is~$O(\|G\|K + \|G\|\lg{\|G\|})$, which is polynomial since~$K$ is polynomially bounded.

    For each ad~$A_i$, the algorithm considers only feasible sequential schedules. If~$A_i$ is in the computed solution, it has exactly~$w_i$ copies in distinct compatible slots. Thus, the algorithm always returns a feasible solution.

    By Lemma~\ref{lema:pd}, any feasible schedule for~$G$ can be modified to obtain a new schedule in which the same ads appear sequentially in the slots. Then, let~$\OPT$ be an optimal schedule for~$G$. There is a schedule~$\OPT'$ in which all ads of~$\OPT$ appears sequentially and~${f(\OPT) = f(\OPT')}$. The schedule~$\OPT'$ induces a feasible solution for the knapsack problem defined since every ad appears sequentially as an item of the knapsack problem with weight~$w_i$ and the capacity of the knapsack is not violated since the algorithm adds at most a copy per slot and~${W = K}$.
    
    A similar argument to the one used by~\citet{kellerer2004introduction} to prove the optimality of the dynamic programming algorithm for the binary knapsack problem can be used to prove that Algorithm~\ref{pd} finds an optimal solution for the knapsack problem with release dates and, thus, for MAXSPACE-R with large ads.
\end{proof}

\subsection{A 1/4-approximation algorithm for medium ads\label{sec:pm}}

In this section, we present an algorithm for ads with medium size. We show this algorithm is a~$1/4$-approximation for ads of~$M$ (Lemma~\ref{prova}).

\begin{algorithm}[H]\scriptsize
\caption{Algorithm for medium ads.\label{alg:2}}
\begin{algorithmic}[1]

\Procedure{Alg\_Medium}{$M$}
    \For{$j \gets 1, \dots, K$}
    \State $B_j \gets \emptyset$
    \EndFor
    \For{\textbf{each }${A_i \in M}$ in non-decreasing order of~$r_i$}\label{lin:goto}
        \State $X \gets \emptyset$
        \For{$k \gets 1, \dots, w_i$} \label{alg2:laco2}
            \If{there exists $j \not \in X$ with $j \geq r_i$ and $B_j$ is empty}\label{lin:caso1}
                \State $j \gets \min\{j' \mid j' \not \in X, \; j' \geq r_i \text{ and } B_{j'} \text{ is empty}\}$
                \State $X \gets X \cup \{j\}$
            \ElsIf{there exists $j \not \in X$ with $j \geq r_i$ and $f(B_j) \leq 1 - s_i$}\label{lin:caso2}
                \State $j \gets \argmin \{f(B_{j'}) \mid j' \not \in X, \; j' \geq r_i\}$
                \State $X \gets X \cup \{j\}$
            \Else
                \State discard $A_i$ and continue at Line~\ref{lin:goto}
            \EndIf
        \EndFor
        \State add a copy of $A_i$ to $B_j$ for each $j \in X$
    \EndFor
    \State \Return $\{B_1, B_2, \dots, B_K\}$
\EndProcedure

\end{algorithmic}
\end{algorithm}

The idea behind Algorithm~\ref{alg:2} is to try to add the ads to the least full compatible slots. It receives as input a set of medium ads~$M$ and iterates over them in order of release date (from smallest to highest). For each copy of an ad~$A_i$, the algorithm finds the first empty slot compatible with~$r_i$. If such a slot exists, the algorithm adds a copy of~$A_i$ to it. Otherwise, the algorithm finds the least full slot compatible with~$r_i$. A set~$X$ is used to maintain the slots to which~$A_i$ was assigned. If it is possible to assign all~$w_i$ copies of~$A_i$, the slots in~$X$ are updated. Otherwise,~the ad $A_i$ is discarded, and the algorithm goes to the next ad. The algorithm returns a schedule of ads to the slots.

Consider the output of Algorithm~\ref{alg:2} and let~${\OPT = \{B^*_1, B^*_2, \dots, B^*_K\}}$ be an optimal schedule. Also, let~$H$ be the set of ads not scheduled by the algorithm and let~$H^*$ be the subset of ads in~$H$ that are in~$\OPT$. In Lemma~\ref{feasible:alg:2}, we show that Algorithm~\ref{alg:2} runs in polynomial time in the instance size and returns a feasible solution. In Lemma~\ref{lema:1}, we show that if there is an ad~$A_i$ in an optimal schedule that was not scheduled by the algorithm, that is,~${A_i \in H^*}$, then each slot~$B_j$ such that~${j \geq r_i}$ has fullness~${f(B_j) \geq 1/4}$ in the solution returned by the algorithm. The Lemma~\ref{lema:1} is used to prove that this algorithm is a~$1/4$-approximation for medium ads (Lemma~\ref{prova}).

\begin{lemma}\label{feasible:alg:2}
    Algorithm~\ref{alg:2} runs in polynomial time in the instance size and returns a feasible solution.
\end{lemma}

\begin{proof}
    Sorting the ads by release date in the loop of Line~\ref{lin:goto} can be done in polynomial time. Finding the slots to assign the ads takes time~$O(K)$. Therefore, the complexity of this algorithm is~${O(K\sum_{A_i \in M}{w_i} + |M|\lg|M|)}$, which is polynomial since~${w_i \le K}$.

    The algorithm adds a copy of an ad~$A_i$ only to compatible slots. Moreover,~$A_i$ is added only if exactly~$w_i$ copies of~$A_i$ can be added to compatible slots. Therefore, the algorithm returns a feasible solution.
\end{proof}

\begin{lemma}\label{lema:1}
    Let~${A_i \in H^*}$ and let~$Z$ be the set of slot indices~$j$ such that~${j \ge r_i}$. For each~${j \in Z}$,~${f(B_j) \geq 1/4}$.
\end{lemma}

\begin{proof}
Consider~${B_1, \dots, B_K}$ at the moment in which the algorithm tries to add an ad~${A_i \in H^*}$.
As~$A_i$ was not added, there exists at least a slot~$B_j$ of~$Z$ whose
fullness is greater than~$1/2$, since~${s_i \leq 1/2}$.
Then, it follows that~$B_j$ has at least~$2$ ads.
Let~$A_{i'}$ be the last ad assigned to~$B_j$ until this moment.
Then, at the moment that~$A_{i'}$ was assigned, the fullness of~$B_j$ was at least~$1/4$ since it had at least one medium ad. Thus, the copy of $A_{i'}$ assigned to~$B_j$ corresponds to the case of Line~\ref{lin:caso2} of the algorithm.
Note that~${r_{i'} \leq r_i}$, by the order in which the algorithm considered the ads. Therefore, it follows that all slots~$B_{j'}$, with~${j' \geq r_i}$, were considered in the case of Line~\ref{lin:caso1} to assign~$A_{i'}$, and no such a slot satisfied this case's criteria, so each of these slots had fullness at least~$1/4$.
\end{proof}

\subsection{A 1/4-approximation algorithm for small ads\label{sec:pp}}

In this section, we present an algorithm which, later on, we prove that it is a~$1/4$-approximation for small ads (Lemma~\ref{prova}).

\begin{algorithm}[H]\scriptsize
\caption{Algorithm for small ads.\label{alg:3}}

\begin{algorithmic}[1]

\Procedure{Alg\_Small}{$P$}
    \For{$j \gets 1, \dots, K$}
        \State $B_j \gets \emptyset$
    \EndFor
    \For{\textbf{each }$A_i \in P$ in non-decreasing order of $r_i$}\label{lin3:goto}
        \State $X \gets \emptyset$
        \For{$k \gets 1, \dots, w_i$}
            \If{there exists $j \not \in X$ with $j \geq r_i$ and $f(B_j) < 1/4$}\label{lin:casoa}
                \State $j \gets \min\{j' \mid j' \not \in X, \; j' \geq r_i \text{ and } f(B_{j'}) < 1/4\}$
                \State $X \gets X \cup \{j\}$
            \ElsIf{there exists ${j_1 \in X}$, ${j_2 \not \in X}$ with $j_1, j_2 \geq r_i$, $f(B_{j_1}) < 1/4$, and $f(B_{j_2}) \geq 3/4$}\label{lin:casob}
                \State $j_1 \gets \min\{j' \mid j' \in X$, $j' \geq r_i$ and $f(B_{j'}) < 1/4\}$
                \State $j_2 \gets \min\{j' \mid j' \not \in X$, $j' \geq r_i$ and $f(B_{j'}) \geq 3/4\}$
                \State find $T \subset B_{j_2}$ such that $1/4 \leq f(T) \leq 1/2$ and $T \cap B_{j_1} = \emptyset$
                \State move $T$ to $B_{j_1}$
                \State $X \gets X \cup \{j_2\}$
            \ElsIf{there exists $j \not \in X$ with $j \geq r_i$ and $f(B_j) \leq 1 - s_i$}\label{lin:casoc}
                \State $j \gets \argmin \{f(B_{j'}) \mid j' \not \in X, j' \geq r_i\}$
                \State $X \gets X \cup \{j\}$
            \Else
                \State discard $A_i$ and continue at Line~\ref{lin3:goto}
            \EndIf
        \EndFor
        \State add a copy of $A_i$ to $B_j$ for each $j \in X$
    \EndFor
    \State \Return $\{B_1, B_2, \dots, B_K\}$

\EndProcedure

\end{algorithmic}
\end{algorithm}

The idea behind Algorithm~\ref{alg:3} is to try to add the ads to the least full compatible slots and, when it is not possible, try to move ads from a slot with high fullness to a slot with low fullness. It receives as input a set of small ads~$P$ and iterates over it in order of release dates, from the smallest to the highest release date.

For each copy of an ad~$A_i$, the algorithm looks up for the first slot with fullness smaller than~$1/4$ compatible with the release date~$r_i$ of~$A_i$. If such a slot exists, the algorithm adds a copy to it. Otherwise, it tries to find two slots~$B_{j_1}$ and~$B_{j_2}$ which are compatible with~$A_i$ and such that~$B_{j_1}$ has a copy of~$A_i$ and fullness smaller than~$1/4$, and~$B_{j_2}$ has no copy of~$A_i$ and has fullness at least~$3/4$. If such slots~$B_{j_1}$ and~$B_{j_2}$ are found, the algorithm moves a subset of ads~$T$ from $B_{j_2}$ to $B_{j_1}$. The set $T$ must have no intersection with~$B_{j_1}$ (that is, the ads of~$T$ do not appear in~$B_{j_1}$) and have fullness of at least~$1/4$ and at most~$1/2$. Note that it is always possible to find such a subset of~$B_{j_2}$. To see this, observe that at least~$1/2$ of the fullness of~$B_{j_2}$ is composed by ads that are not in~$B_{j_1}$ since the fullness of $B_{j_1}$ is at most $1/4$. From this subset with fullness at least $1/2$, it is possible to find a subset of fullness of at least~$1/4$ and at most~$1/2$, since the ads are small. The algorithm then moves~$T$ from~$B_{j_2}$ to~$B_{j_1}$ and add a copy of~$A_i$ to~$B_{j_2}$. Note that this movement does not violate any restriction of release dates since the ads of~$T$ have release dates at most~$r_i$. When no such a pair~$B_{j_1}$ and~$B_{j_2}$ is found, the algorithm tries to add a copy of~$A_i$ to the first slot where it fits.

A set~$X$ is used to maintain the slots to which~$A_i$ was assigned. If it is possible to assign all~$w_i$ copies of~$A_i$, the slots in~$X$ are updated. Otherwise,~$X$ is ignored, and the algorithm goes to the next ad. The algorithm returns a schedule of ads to the slots.

Consider the output of Algorithm~\ref{alg:3} and let~${\OPT = \{B^*_1, B^*_2, \dots, B^*_K\}}$ be an optimal schedule. Also, let~$H$ be the set of ads not scheduled by the algorithm and let~$H^*$ be the subset of ads in~$H$ that are in~$\OPT$. In Lemma~\ref{feasible:alg:3}, we show that Algorithm~\ref{alg:3} runs in polynomial time in the instance size and returns a feasible solution. In Lemma~\ref{lema:2}, we show that if there exists some ad~$A_i$ that is in an optimal schedule but was not scheduled by the algorithm, that is,~${A_i \in H^*}$, then each slot~$B_j$ such that~${j \geq r_i}$ has fullness~${f(B_j) \geq 1/4}$ in the solution returned by the algorithm. The Lemma~\ref{lema:2} is used to prove that this algorithm is a~$1/4$-approximation for small ads (Lemma~\ref{prova}).

\begin{lemma}\label{feasible:alg:3}
    Algorithm~\ref{alg:3} runs in polynomial time in the instance size and returns a feasible solution.
\end{lemma}

\begin{proof}
    Sorting the ads by release date in the loop of Line~\ref{lin3:goto} can be done in polynomial time. Finding the slots to assign the ads takes time~$O(K)$, and changing the ads from slots takes time~$O(K|P|)$. Therefore, the complexity of this algorithm is~${O(K|P|\sum_{A_i \in P}{w_i} + |P|\lg|P|)}$, which is polynomial since~${\sum_{A_i \in P}{w_i} \in O(K|P|)}$.

    The algorithm adds a copy of an add~$A_i$ only to compatible slots. Besides that,~$A_i$ is scheduled only if exactly~$w_i$ copies of~$A_i$ can be added to compatible slots. When the algorithm moves a set of ads from a slot~$B_{j_2}$ to a slot~$B_{j_1}$, it does not violate any restriction of release dates since the ads in~$T$ have release dates smaller than or equal to the release date of the ad considered by the iteration, by the order in which the ads are considered. Also, the algorithm does not violate the fullness of any slot since ${f(B_{j_1}) < 1/4}$ and~${f(T) \leq 1/2}$. The set of ads~${T \subset B_{j_2}}$ is not in~$B_{j_1}$, then the restriction of each ad has at most a copy per slot is also not violated when the algorithm moves~$T$ from~$B_{j_2}$ to~$B_{j_1}$. Therefore, the algorithm returns a feasible solution.
\end{proof}

\begin{lemma}\label{lema:2}
    Let~${A_i \in H^*}$ and let~$Z$ be the set of indices~$j$ such that~${j \ge r_i}$. For each~${j \in Z}$,~${f(B_j) \geq 1/4}$.
\end{lemma}


\begin{proof}
Consider the slots~${B_1, \dots, B_K}$ at the moment the algorithm tries to assign~$A_i$. Consider the moment in which ad~$A_i$ was discarded. Since the case of Line~\ref{lin:casoa} fails, all slots~${Z \setminus X}$ have fullness at least~$1/4$. And, as the case of Line~\ref{lin:casoc} fails, there exists at least one~${j \in Z \setminus X}$ with~${f(B_j) > 3/4}$ (since~${s_i < 1/4}$). Finally, since the case of Line~\ref{lin:casob} fails, all slots in~$X$, at this moment, had fullness at least~$1/4$. Note that the fullness of $j_1$ remains at least~$1/4$ after the ads of~$T$ are removed from $B_{j_1}$. We conclude that, at this moment, all slots in~${(Z \setminus X) \cup X = Z}$ had fullness at least~${1/4}$. Therefore, the result follows.
\end{proof}

\subsection{A 1/9-approximation algorithm for the general case\label{sec:pcasogeral}}

Now, we present an algorithm for the general case, showing that it is a~$1/9$-approximation. First, we show on Lemma~\ref{prova} that Algorithms~\ref{alg:2} and~\ref{alg:3} are~$1/4$-approximations for, respectively, medium and small ads. Then, in Algorithm~\ref{alg:5}, we present a pseudocode for the whole set of ads~$\cala$.

\begin{lemma}\label{prova}
    Algorithms~\ref{alg:2} and~\ref{alg:3} are~$1/4$-approximations for medium and small ads, respectively.
\end{lemma}

\begin{proof}

Consider the execution of one of these algorithms.
Let~$W$ be the set of copies of ads (considering~$w_i$ copies of each ad~$A_i$) scheduled by the algorithm.
Let~$W^*$ be the ads scheduled in an optimal solution~$\OPT$ (also considering~$w_i$ copies of each ad~$A_i$), such that~${f(\OPT) = f(W^*)}$. Also, let~${B^*_j \subseteq W^*}$ be the ads scheduled in a slot~$B_j$ in~$\OPT$.
We partition~$W^*$ into~${E^* = W^* \cap W}$ and~${N^* = W^* \setminus W}$. The set~$E^*$ corresponds to the ads in~$W^*$ scheduled by the algorithm, and the set~$N^*$ to those not scheduled. Let~$\ell$ be smallest index of slot such that~${B^*_\ell \cap N^* \neq \emptyset}$. Let~$m$ be the smallest index of slot such that~${f(B_j) \geq 1/4}$ for all~${j \geq m}$. Note that~${m \leq \ell}$, since each slot~$B_j$ with~${j \ge \ell}$ has fullness~${f(B_j) \ge 1/4}$, by Lemmas~\ref{lema:1} and~\ref{lema:2}. This implies that each ad~${A_i \in B^*_j}$ with~${j \leq m}$ is in~$E^*$ by the minimality of~$\ell$.

Let~$Z = \{1, 2, \dots, m-1\}$ and~$\overline{Z} = \{m, m+1, \dots, K\}$.
Let~${H = \bigcup_{j \in Z} B^*_j}$. Let~$\tilde{w}_i$ be the number of copies of an ad~$A_i$ scheduled by the algorithm in slots with index in~$Z$, and let~$\tilde{w}^*_i$ be the number of copies of~$A_i$ scheduled in slots with index in~$Z$ in the optimal solution.
Let~$F$ be the set of ads which have been scheduled by the algorithm in slots of~$Z$ at least as many times as the optimal solution does,
and let~$R$ be the set of ads that have been scheduled by the algorithm in slots of~$Z$ fewer times than the optimal solution does, that is,~${F = \{A_i \mid \tilde{w}_i \ge \tilde{w}^*_i\}}$ and~${R = H \setminus F}$. Let~$Q$ be the set of indices~$j$ of~$Z$ with~${f(B_j) \ge 1/4}$ in the solution computed by the algorithm.

Let~${A_i \in R}$ and~${j \in Z \setminus Q}$ with~${j \ge r_i}$. We will prove that~${A_i \in B_j}$ in the solution computed by the algorithm. Since~${A_i \in R}$, there is a copy of~$A_i$ in~$\overline{Z}$. Assume that~${A_i \not \in B_j}$, then the algorithm did not add a copy of~$A_i$ to~$B_j$ because~${f(B_j) \ge 1/4}$, then~${j \in Q}$, which is a contradiction. We conclude that~${A_i \in B_j}$ in the algorithm's solution.

We are going to prove that~${f(R) < 1/4}$. First, note that~${m-1 \notin Q}$ by the minimality of~$m$. Now, observe that~${R \subseteq B_{m-1}}$. In fact, if~${A_i \in R}$, then~${r_i \le m-1}$ and by the previous paragraph we know that~${A_i \in B_{m-1}}$. It follows that $R \subseteq B_{m-1}$, and then~${f(R) \le f(B_{m-1}) < 1/4}$.

To derive the lemma, it suffices to show that $\sum_{j \in Z}{f(B_j)} \ge 1/4\sum_{j \in Z}{f(B^*_j)}$. We can rewrite these sums as follows:
\begin{align*}
  \sum_{j \in Z}{f(B^*_j)} = \sum_{A_i \in R}{s_i\tilde{w}^*_i} + \sum_{A_i \in F}{s_i\tilde{w}^*_i}\qquad \text{and} \qquad
  \sum_{j \in Z}{f(B_j)} \ge \sum_{A_i \in R}{s_i\tilde{w}_i} + \sum_{A_i \in F}{s_i\tilde{w}_i}.
\end{align*}
The equality follows because the ads scheduled by the optimal solution in slots of~$Z$ correspond to~$H$, partitioned by~${R, F}$. The inequality
follows from the definition of~$w_i$.

Note that~${\sum_{A_i \in F}{s_i\tilde{w}_i} \ge \sum_{A_i \in F}{s_i\tilde{w}^*_i}}$. Thus, if~$\sum_{A_i \in F}{s_i\tilde{w}^*_i} \ge 1/4\sum_{j \in Z}{f(B^*_j)}$, the statement follows. Then, in the following we assume that~$\sum_{A_i \in F}{s_i\tilde{w}^*_i} < 1/4\sum_{j \in Z}{f(B^*_j)}$, which implies that~$\sum_{A_i \in R}{s_i\tilde{w}^*_i} \ge 3/4 \sum_{j \in Z}{f(B^*_j)}$.

The fullness of slots in~$Z$ in the solution found by the algorithm can be rewritten as:
\begin{align*}
  \sum_{j \in Z}{f(B_j)}
    &=  \sum_{j \in Q}{f(B_j)} + \sum_{j \in Z \setminus Q}{f(B_j)}\\
    &\ge \sum_{j \in Q}{\frac{1}{4}} + \sum_{j \in Z \setminus Q}{f(B_j)}\\
    &> \sum_{j \in Q}{f(R)} + \sum_{j \in Z \setminus Q}{f(B_j \cap R)}\\
    &\ge \sum_{A_i \in R}{(m - r_i)s_i}\\
    &\ge \sum_{A_i \in R}{\tilde{w}^*_i s_i} \ge \frac{3}{4} \sum_{j \in Z}{f(B^*_j)}.
\end{align*}
The first inequality holds by the definition of~$Q$. The second inequality holds because~${f(R) < 1/4}$. For the third one, consider the sums on the left side of the inequality and notice that each ad~${A_i \in R}$ appears in all terms of the first sum and in all terms of the second sum of indices~$j$ with~${j \ge r_i}$; thus,~$A_i$ appears in at least~$(m-r_i)$ terms.
The penultimate inequality holds because an ad cannot be displayed before the release date. Thus, also in this case we conclude that~${\sum_{j \in Z}{f(B_j)} \ge 1/4\sum_{j \in Z}{f(B^*_j)}}$.

Finally, we bound the value of the solution~$W$:
\begin{align*}
f(W)
    & = \sum_j f(B_j) = \sum_{j \in Z} f(B_j) + \sum_{j \in \overline{Z}}{f(B_j)}\\
    & \geq \sum_{j \in Z} {\frac{1}{4}f(B^*_j)} + \sum_{j \in \overline{Z}}{\frac{1}{4}}\\
    & \geq \sum_{j \in Z} {\frac{1}{4}f(B^*_j)} + \sum_{j \in \overline{Z}}{\frac{1}{4} f(B^*_j)}\\
    & = \frac{1}{4}\left(\sum_{j \in Z} {f(B^*_j)} + \sum_{j \in \overline{Z}}{f(B^*_j)}\right)\\
    & = \frac{1}{4}f(\OPT).
\end{align*}

The first inequality holds by the definition of~$m$ and the statement of the previous paragraph. The second inequality holds by the fact that ${\sum_{j \in Z}{f(B^*_j)} \leq 1}$.
\end{proof}

\begin{algorithm}[H]\scriptsize
\caption{Algorithm for general case\label{alg:5}}
\begin{algorithmic}[1]

\Procedure{Alg\_General}{$\cala$}
    \State ${G = \{A_i \in \cala \mid s_i > 1/2\}}$
    \State ${M = \{A_i \in \cala \mid 1/4 < s_i \leq 1/2\}}$
    \State ${P = \{A_i \in \cala \mid s_i \leq 1/4\}}$
    \State $S_1 \gets \Call{DP}{G}$
    \State $S_2 \gets \Call{Alg\_Medium}{M}$
    \State $S_3 \gets \Call{Alg\_Small}{P}$
    \State \Return $\max\{S_1, S_2, S_3\}$
\EndProcedure

\end{algorithmic}
\end{algorithm}

The Algorithm~\ref{alg:5} divides the ads into large~$G$, medium~$M$ and small~$P$ and executes the Algorithms~\ref{pd},~\ref{alg:2} and~\ref{alg:3}, respectively, for~$G$,~$M$ and~$P$. Finally,
the algorithm returns the best of the solutions of the executed algorithms. In Theorem~\ref{theorema:1}, we show that this algorithm is a~$1/9$-approximation for MAXSPACE-R.





\begin{theorem}\label{theorema:1}
    Algorithm~\ref{alg:5} is a~$1/9$-approximation for MAXSPACE-R problem.
\end{theorem}


\begin{proof}
    Algorithm~\ref{alg:5} only partitions the ads and executes Algorithms~\ref{pd},~\ref{alg:2} and~\ref{alg:3}. By Lemmas~\ref{otimo:pd},~\ref{feasible:alg:2} and~\ref{feasible:alg:3}, these algorithms run in polynomial time in the instance size and return feasible solutions. Then, Algorithm~\ref{alg:5} runs in polynomial time in the instance size and returns a feasible solution.

    Let~$W^*$ be the copies of ads scheduled in an optimal solution (considering~$w_i$ copies of each ad~$A_i$) and let~${f(\OPT)= f(W^*)}$. If~$f(W^* \cap G) \geq \frac{1}{9}f(\OPT)$, it is possible to obtain a solution with fullness at least~$\frac{1}{9}f(\OPT)$, since Algorithm~\ref{pd} is exact for large ads (Lemma~\ref{otimo:pd}). Otherwise, we know that~${f(W^* \cap M) \geq \frac{4}{9}f(\OPT)}$ or ${f(W^* \cap P) \geq \frac{4}{9}f(\OPT)}$. If~${f(W^* \cap M) \geq \frac{4}{9}f(\OPT)}$, then a solution for ads of~$M$ has fullness at least $\frac{1}{4} (\frac{4}{9} f(W^*)) = \frac{1}{9}f(\OPT)$, since Algorithm~\ref{alg:2} is a~$1/4$-approximation for medium ads (by Lemma~\ref{prova}). If~${f(W^* \cap P) \geq \frac{4}{9}f(\OPT)}$, then a solution for ads of~$P$ has fullness at least~${\frac{1}{4} (\frac{4}{9} f(W^*)) = \frac{1}{9}f(\OPT)}$, since Algorithm~\ref{alg:3} is a~$1/4$-approximation for small ads (by Lemma~\ref{prova}).
\end{proof}

\section{A PTAS for MAXSPACE-RDV with a constant number of slots\label{sec:ptasrd}}

In what follows, assume that the number of slots~$K$ is a constant, ${L = 1}$, and~${0 < s_i \le 1}$ for each~$A_i \in \cala$. In MAXSPACE-RDV, we define~$f(B_j) = \sum_{A_i \in B_j}{v_i}$ as the value of a slot~$B_j$ and~$f(S) = \sum_{B_j \in S}{f(B_j)}$ as the value of a solution~$S$.

Let~$S$ denote a feasible solution~${\cala' \subseteq \cala}$ scheduled into slots~${B_1, B_2, \dots, B_K}$. The \textit{type}~$t$ of an ad~$A_i \in \cala'$ with respect to~$S$ is the subset of slots to which~$A_i$ is assigned, that~is,~${A_i \in B_j}$ if and only if~${j \in t}$. Let~$\calt$ be a set of all the subsets of slots, then~$\calt$ contains every possible type and~${|\calt| = 2^K}$.
Observe that two ads with the same type have the same frequency, and thus one can think of all ads in~$\cala'$ with the same type as a single ad. 

Let~${\varepsilon > 0}$ be a constant such that~$1/\varepsilon$ is an integer, and let ${q = \min\{|\cala|, \const/\varepsilon \}}$. Our algorithm guesses a set $V$ with at most~$q$ ads with the largest values in an optimal solution. For each~$V \subseteq \cala$ such that $|V| \leq q$, we define~$U$ as the set of every ad~${A_i \in \cala \setminus V}$ such that~$v_i w_i \leq \Vmin$, where~$\Vmin = \min\{v_i w_i : A_i \in V\}$. Then, for each feasible scheduling of $V \subseteq \cala$, we fill the remaining spaces in the slots with ads in $U$ using a linear program.

A \textit{configuration} for a subset of ads~$\cala' \subseteq \cala$ is a feasible solution which schedules every ad in~$\cala'$.
Lemma~\ref{lemma:0} states that if~$K$ is constant, then the number of possible configurations containing only ads in~$V$ is polynomial in the number of ads in $V$ and can be enumerated by a brute-force algorithm.

\begin{lemma}\label{lemma:0}
    If~$K$ is constant, then the configurations for all subsets of~$V$ can be listed in polynomial time.
\end{lemma}

\begin{proof}
  There are $O(q|\cala|^q)$ possible choices for~$V$, and there exist~$O(q^{2^K})$ possible solutions for each set since the number of types is $2^K$. Thus, we can enumerate $V$ and all of its configurations in time
  $O(q^{2^K + 1}|\cala|^q)$, which is polynomial since~$K$ and~$q$ are constants.
\end{proof}

Since all candidate configurations can be listed in polynomial time by Lemma~\ref{lemma:0}, we may assume that we guessed the configuration~$S_V$ of the most valuable ads induced by~$\OPT_V$. We are left with the residual problem of placing ads of~$U$.
For each slot~$j$ in $S_V$,~${1 \le j \le K}$, the space which is unused by the most valuable ads~$V$ is
\[
  u_j = 1 - \sum_{A_i \in B_j} s_i.
\]
We define RESIDUAL-MAXSPACE-RDV as the problem of, given a set of ads~$U$, where~${v_i w_i \leq \Vmin}$ for all~${A_i \in U}$, finding a subset~${\cala'_t \subseteq U}$ for each~${t \in \calt}$ such that the occupation of each slot~$j$ is at most~$u_j$, and which maximizes the value
\[
\sum_{t \in \calt} \sum_{A_i \in \cala'_t} |t| v_i.
\]

 Let $\calt(A_i)$ be the subset of types \textit{compatible} with an ad~$A_i$, i.e., the set of types~${t \in \calt}$ with~${|t| = w_i}$, and such that the slots in~$t$ are compatible with the release date~$r_i$ and the deadline~$d_i$. We solve the linear program~(P) to assign ads of~$U$ to types.
\begin{alignat}{4}
    (P) &&
    \text{Maximize } & \sum\limits_{A_i \in U}~\sum\limits_{t \in \calt(A_i)}{v_i w_i \MF_{A_i, t}} & \label{eq1}\\
    &&\text{Subject to: }
    & \sum\limits_{t \in \calt(A_i)}{\MF_{A_i, t}} \leq 1 & & \forall{A_i \in U}\label{eq2} \\
    &&& \sum\limits_{\stackrel{t \in \calt\colon}{j \in t}}~\sum\limits_{A_i \in U}{s_i\MF_{A_i, t}} \leq u_j & \qquad & j = 1, 2, \dots, K\label{eq3} \\
    &&& \MF_{A_i, t} \ge 0 & & \forall A_i \in U, \forall t \in \calt(A_i)
\end{alignat}

The variables~$\MF_{A_i, t}$ indicate if ad~$A_i$ is assigned to type~$t$, constraints~\eqref{eq2} ensure that an ad cannot be assigned more than once, and constraints~\eqref{eq3} guarantee that the capacity of any slot will not be violated.

Consider a solution~$\MF$ for~(P), which can be obtained in polynomial time~\citep{karmarkar1984new}, and notice that $\MF$ induces an assignment of ads to types. In this assignment, if the solution is such that~$\MF_{A_i,t} < 1$ units from ad~$A_i$ are assigned to type~$t$, then we say that ad~$A_i$ is fractionally assigned to~$t$ by an amount of~${\MF_{A_i,t}}$. The set of all types~$t$ for which~${\MF_{A_i,t} > 0}$ is called the \textit{support} of~$A_i$ and is denoted by~$\SUP(A_i)$.

To eliminate fractional assignments, we group ads with the same support. Let~$\WW$ be a subset of types, and~$\PW$ be the set of ads~$A_i$ with~${\SUP(A_i) = \WW}$. In particular, each ad~${A_i \in \PW}$ is compatible with any type~${t \in \WW}$. For each type~${t \in \WW}$, we define the total fullness received by~$t$ from~$\PW$ as
\[
z_t = \sum_{A_i \in \PW} s_i \MF_{A_i, t}.
\]
By the fact that each ad in~$\PW$ is fractionally assigned to types in~$\WW$, we know that
\[
\sum_{A_i \in \PW} s_i \ge \sum_{A_i \in \PW} \sum_{t\in \WW} s_i \MF_{A_i,t}
=
\sum_{t \in \WW} z_t.
\]
In other words, the total size of~$\PW$ given by~${\sum_{A_i \in \PW}{s_i}}$ is not smaller than the size received by types~$\WW$ from~$\PW$.
Therefore, we remove the fractional assignment of all ads in~$\PW$ and integrally reassign each ad in~$\PW$ to types in~$\WW$, discarding any remaining ad.

The process of rounding the fractional assignment is summarized in Algorithm~\ref{alg:1k}, which receives as input a fractional assignment~$\MF$ of ads to types~$\WW$ and returns an integer assignment~$\MF'$.

As part of the rounding, we use a procedure called $\Call{Reassign}{}$, which takes a support~$\WW$ and the ads scheduled to that support~$\PW$ and returns a new allocation of these ads in~$\WW$. This procedure removes all ads from $\PW$ from the fractional solution and greedily fills their space with ads from $\PW$ in order of efficiency ($v_i/s_i$). Note that this new schedule does not have a worse value since all the space is filled, the advertisements are chosen in order of efficiency, and this new solution is fractional. Let $\PWL \subseteq \PW$ be the newly scheduled ads; note that all $\PWL$ ads, except perhaps the last one, are fully scheduled to $\WW$ types. The $\Call{Reassign}{}$ pseudocode is presented in Algorithm~\ref{alg:ra}.

\begin{algorithm}[t]\scriptsize
    \caption{Algorithm for reassigning fractional solution.\label{alg:ra}}
    \begin{algorithmic}[1]
    \Procedure{Reassign}{$\WW, \PW$}
        \For{\textbf{each} $A_i \in \PW$ and $t \in \WW$}
            \State $\MF'_{A_i, t} \gets 0$
        \EndFor
        \For{\textbf{each} $t \in \WW$}
            \State $z_t \gets \sum_{A_i \in \PW} s_i \MF_{A_i, t}$
        \EndFor
        \State $\PWL \gets \emptyset$
        \State $z_{\WW} \gets \sum_{t \in \WW}z_t$ \Comment{Total area of items with support~$\WW$}
        \For{\textbf{each} $A_i \in \PW$ in non-increasing order of~$v_i/s_i$}\label{linha:ra8}
            \If{$\sum_{A_j \in \PWL}{s_j} < z_{\WW}$} \Comment{The last ad may not fit entirely}
                \State $\PWL \gets \PWL \cup \{A_i\}$
            \EndIf
        \EndFor
        \State Let~$\{t_1, t_2, \dots, t_{|\WW|}\}$ be the types in~$\WW$
        \State $k \gets 1$
        \For{\textbf{each} $A_i \in \PWL$}
            \State $s_i' \gets s_i$
            \While{$s_i' > 0$}
                \State $m \gets \min\{s_i', z_{t_k} - s_i'\}$
                \State $\MF'_{A_i, t_k} \gets m / s_i$
                \State $s_i' \gets s_i' - m$
                \State $z_{t_k} \gets z_{t_k} - m$
                \If{$z_{t_k} = 0$}
                    \State $k \gets k + 1$
                \EndIf
            \EndWhile
        \EndFor
        \State \Return $\MF', \PWL$
    \EndProcedure
    \end{algorithmic}
\end{algorithm}

In Lemma~\ref{lemma:rapoly}, we observe that Algorithm~\ref{alg:ra} is polynomial in the instance size. Lemma~\ref{lemma:ravalue} shows that \Call{Reassign}{} does not worsen the solution.

\begin{lemma}\label{lemma:rapoly}
    Algorithm~\ref{alg:ra} runs in polynomial time.
\end{lemma}
\begin{proof}
  The size of~$\PW$ is~$O(n)$ and the size of~$\WW$ is~$O(2^K)$; thus, Algorithm~\ref{alg:ra} running time is~$O(n2^K)$, which is polynomial since~$K$ is constant.
\end{proof}

\begin{lemma}\label{lemma:ravalue}
    Algorithm~\ref{alg:ra} returns a solution with the same value of linear programming~(P) assignment.
\end{lemma}
\begin{proof}
    The algorithm fills the space for the $\PW$ items in the $\WW$ types with the best efficiency items in $\PW$. This way, the algorithm obtains an optimal value for the fractional allocation of the $\PW$ items in the considered space. The linear programming algorithm~(P) also obtains a fractional optimal solution for the same items and considers the same space. Therefore, both solutions have the same value.
\end{proof}

\begin{algorithm}[H]\scriptsize
\caption{Algorithm for rounding ad assignment.\label{alg:1k}}
\begin{algorithmic}[1]
\Procedure{Rounding}{$\MF$}
    \For{\textbf{each} $A_i \in U$ and $t \in \calt$}
        \State $\MF'_{A_i, t} \gets 0$
    \EndFor
    \For{\textbf{each} $\WW \subseteq \calt$}\label{linha:1}
        \State $\PW \gets$ all ads $A_i$ with $\SUP(A_i) = \WW$
        \State $\MF', \PWL \gets \Call{Reassign}{\WW, \PW}$
        \State discard from~$\PWL$ any ad that is not integrally assigned to the same type in~$\WW$\label{lin:7}
    \EndFor
    \State \Return $\MF'$
\EndProcedure

\end{algorithmic}
\end{algorithm}

Lemma~\ref{lemma:1} shows that Algorithm~\ref{alg:1k} is polynomial in the instance size. Corollary~\ref{corolario:0} is obtained from Lemma~\ref{lemma:2}, and bounds the total value of ads discarded by Algorithm~\ref{alg:1k} in each execution of Line~\ref{lin:7}. And Corollary~\ref{corolario:1} is obtained from Lemma~\ref{lemma:2} and Corollary~\ref{corolario:0}.

\begin{lemma}\label{lemma:1}
    Algorithm~\ref{alg:1k} runs in polynomial time in the instance size.
\end{lemma}
\begin{proof}
  The loop of Line~\ref{linha:1} executes a constant number of iterations, since~${|\calt| = 2^K}$ and the number of subsets of~$\calt$ is~$2^{2^K}$. The $\Call{Reassign}{}$ algorithm is also polynomial, by Lemma~\ref{lemma:rapoly}. Then, the algorithm runs in polynomial time.
\end{proof}

\begin{lemma}\label{lemma:2}
  Let~${\WW \subseteq \calt}$ and let~$\PWL$ be the set of ads with support~$\WW$ after the reassign algorithm.
  The number of discarded ads from~$\PWL$ in Line~\ref{lin:7} of the $\Call{Rounding}{}$ is at most~$|\WW|$.
\end{lemma}
\begin{proof}
    The $\Call{Reassign}{}$ algorithm adds an advertisement fractionally to a type in two ways: starting at a type $t$ and continuing at a type $t+1$, and completing the fullness in the last type of $\WW$. In the first case, the algorithm can add a maximum of $|\WW| - 1$ fractional ads, and in the second case, it is possible to add at most one ad fractionally to a type. Thus, at most $|\WW|$ advertisements are added fractionally to types of $\WW$, and the result follows.
\end{proof}

\begin{corollary}\label{corolario:0}
  Let~${\WW \subseteq \calt}$ and let~$\PWL$ be the set of ads scheduled to~$\WW$ after reassigning the algorithm.
  Then the total value of selected ads in~$\PWL$ after the execution of rounding is
  \[
  \sum_{A_i \in \PWL}\sum_{t \in \WW} v_i w_i \MF'_{A_i,t} \ge \sum_{A_i \in \PWL}\sum_{t \in \WW} v_i w_i \MF_{A_i,t} - |\WW|\Vmin.
  \]
\end{corollary}
\begin{proof}
Let~$\PWLL$ be the set of discarded advertisements of~$\PWL$,
\begin{align*}
  \sum_{A_i \in \PWL}\sum_{t \in \WW} v_i w_i \MF'_{A_i,t}&
  \ge \sum_{A_i \in \PWL}\sum_{t \in \WW} v_i w_i \MF_{A_i,t} - \sum_{A_i \in \PWLL}{v_i w_i}\\
  &\ge\sum_{A_i \in \PWL}\sum_{t \in \WW} v_i w_i \MF_{A_i,t} - |\WW|\Vmin.
\end{align*}

The second inequality holds because the number of ads discarded in~$\PW$ is at most~$|\WW|$ in Line~\ref{lin:7} (Lemma~\ref{lemma:2}), and all advertisements~$A_i \in \PWL$ has value~$v_i w_i \leq \Vmin$, by the definition of~$U$.
\end{proof}

\begin{corollary}\label{corolario:1}
The difference between the maximum fractional and modified solution values is not larger than $\const\Vmin$. That is,
\[
\sum_{A_i \in \cala} \sum_{t \in \calt} v_i w_i\MF'_{A_i,t}  \ge \sum_{A_i \in \cala}\sum_{t \in \calt} v_i w_i\MF_{A_i,t} - \const\Vmin.
\]
\end{corollary}
\begin{proof}
Consider the value of variables~$\WW$ and $\PWL$ of Algorithm~\ref{alg:1k}.
Using Corollary~\ref{corolario:0}, we have that
\begin{align*}
\sum_{A_i \in \cala} \sum_{t \in \calt} v_i w_i \MF'_{A_i,t}
&=
\sum_{\WW\subseteq\calt}\sum_{A_i\in \PWL} \sum_{t \in \WW} v_i w_i\MF'_{A_i,t} \\
&\ge \sum_{\WW\subseteq\calt}\left(
                           \sum_{A_i \in \PWL} \sum_{t \in \WW} v_i w_i\MF_{A_i,t}
                           - |\WW|\Vmin
                           \right)\\
&\ge\sum_{A_i \in \cala}\sum_{t \in \calt} v_i w_i\MF_{A_i,t} -
                         \sum_{\WW\subseteq\calt} 2^K\Vmin\\
&=\sum_{A_i \in \cala}\sum_{t \in \calt} v_i w_i\MF_{A_i,t} - \const\Vmin,
\end{align*}
where the last inequality holds because~${|\WW| \leq 2^K}$, and the last equality holds because there are~${2^{|\calt|} = 2^{2^K}}$ distinct choices for~$\WW$.
\end{proof}

The complete algorithm for MAXSPACE-RDV is presented in Algorithm~\ref{alg:2k}. Given parameter $\varepsilon>0$, this algorithm receives a set of ads~$\cala$ as input. The algorithm tries to guess which~$q$ ads are most valuable for an optimal solution. It explores all possible combinations of subsets~$V \subseteq \cala$ with at most~$q$ ads, and for each feasible scheduling for~$V$, it tries to fill the remaining spaces with ads less valuable than the ones in~$V$, called~$U$. In this step, the algorithm associates ads of~$U$ to types using the linear program~(P). The Algorithm $\Call{Rounding}{}$ transforms the fractional assignment~$\MF$ into an integer assignment~$\MF'$. Note that this assignment can be easily converted into a schedule of ads~$U$ into solution~$S'$. The algorithm returns the solution of maximum value among those considered.


\begin{algorithm}[H]\scriptsize
\caption{Algorithm for MAXSPACE-RDV with~$K$ constant.\label{alg:2k}}
\begin{algorithmic}[1]

\Procedure{AlgRDV$_\varepsilon$}{$\cala$}
    \State $q \gets \min\{|\cala|, \const/\varepsilon\}$
    \State $S \gets \emptyset$
    \For{\textbf{each} $V \subseteq \cala$ \textbf{ such that } $|V| \le q$}\label{line6:3}
        \For{\textbf{each} feasible assignment~$S_V$ of~$V$}\label{line6:4}
            \State $\Vmin \gets \min\{v_i w_i: A_i \in V\}$
            \State $U \gets \{A_i \in \cala \setminus V \mid v_i w_i \leq \Vmin\}$
            \State $\MF \gets $ solve LP~(P) with ads in~$U$
            \State $\MF' \gets \Call{Rouding}{\MF}$
            \State Add ads of~$U$ to $S_U$ according to integral assignment $\MF'$
            \State $S' \gets S_V \cup S_U$
            \If{$f(S') \geq f(S)$}
                \State $f(S) \gets f(S')$
            \EndIf
        \EndFor
    \EndFor
    \State \Return $S$
\EndProcedure

\end{algorithmic}
\end{algorithm}

In Lemma~\ref{lemma:3} and~\ref{lemma:8}, we prove that Algorithm~\ref{alg:2k} is polynomial in the instance size and returns a feasible solution. In Theorem~\ref{theorem:0}, we prove that Algorithm~\ref{alg:2k} is a PTAS for MAXSPACE-RDV\@.

\begin{lemma}\label{lemma:3}
  Algorithm~\ref{alg:2k} executes in polynomial time.
\end{lemma}
\begin{proof}
  The linear program is solved in polynomial time in the size of the model~\citep{karmarkar1984new}, and the model is polynomial in the size of the instance since it has $O(|U| + K)$ restrictions and~$O(|U|2^{2^K})$ variables. The~$\Call{Rouding}{}$ algorithm is also polynomial, by Lemma~\ref{lemma:1}. The loops on Lines~\ref{line6:3} and~\ref{line6:4} are polynomial, by Lemma~\ref{lemma:0}. Then, Algorithm~\ref{alg:2k} is polynomial in the instance size.
\end{proof}

\begin{lemma}\label{lemma:8}
    Algorithm~\ref{alg:2k} returns a feasible solution.
\end{lemma}

\begin{proof}
  Since each ad configuration in~$V$ is feasible,~$S_V$ respects release date and deadline restrictions. Solution $S_U$ also respects the release date and deadline restrictions, and constraints~\eqref{eq3} guarantee that this solution respects the slots' capacities. Thus, the algorithm returns a feasible solution.
\end{proof}

\begin{theorem}\label{theorem:0} 
    Algorithm~\ref{alg:2k} is a PTAS for MAXSPACE-RDV\@.
\end{theorem}

\begin{proof}
  We try every schedule for~$V$ with~$|V| \le q$. Thus, consider the moment when the schedule of~$V$ is the same as the~$|V|$ most valuable ads in an optimal solution~$\OPT$.
  Let~$S_V$ be the schedule of ads of~$V$ in the returned solution~$S$. Thus,~$f(S_V) = f(\OPT_V)$, where~$\OPT_V$ is the schedule of~$V$ in~$\OPT$. Note that, if~$q = |\cala|$ or $|\OPT| \le q$, then~$f(S) = f(S_V) = f(\OPT_V) = f(\OPT)$ and the result follows. Now, consider that~${q = \const/\varepsilon < |\cala|}$ and $|\OPT| > q$.

  Let~$\MF$ be the linear program solution and~$\MF'$ be the output of~$\Call{Rounding}{}$. Define
  \[
  f(\MF) = \sum_{A_i \in \cala} \sum_{t \in \calt} v_i w_i\MF_{A_i,t}
  \quad\mbox{ and }\quad
  f(\MF') = \sum_{A_i \in \cala} \sum_{t \in \calt} v_i w_i\MF'_{A_i,t}.
  \]
  Let~$\OPT_U$ be an optimal solution for ads in~$U$ in the remaining spaces of~$\OPT_V$. Observe that~$\OPT_U$ induces a feasible solution with value~$f(\OPT_U)$. This implies that~${f(\MF) \ge f(\OPT_U)}$, as~$\MF$ is an optimal fractionally solution in the remaining spaces of~$S_V$, which has the same fullness of~$\OPT_V$.
  Also, note that~${f(S) = f(\MF') + f(S_V)}$, then using Corollary~\ref{corolario:1} we have
  \begin{align*}
  f(S) &= f(\MF') + f(S_V) \\
  &= f(\MF') + f(\OPT_V) \\
  &\ge f(\MF) - \const\Vmin + f(\OPT_V) \\
  &\ge f(\OPT_U) - \const\frac{f(S_V)}{q} + f(\OPT_V) \\
  &= f(\OPT_U) - \const\frac{f(S_V)}{\frac{\const}{\varepsilon}} + f(\OPT_V) \\
  &= f(\OPT_U) - \varepsilon{f(S_V)} + f(\OPT_V)\\
  &\geq f(\OPT) - \varepsilon{f(\OPT)}.
  \end{align*}
  Where the first inequality holds by Corollary~\ref{corolario:1}, the second inequality holds since~${\Vmin \le f(S_V)/q}$ and the last inequality holds since~$f(\OPT) \geq f(S_V)$.
    
  Since the algorithm returns the best solution and considers solution~$S$, the result follows. \qedhere
  \renewcommand{\qed}{}
\end{proof}

\section{Final remarks\label{sec:con}}

This paper consider two generalizations for the MAXSPACE problem, called MAXSPACE-R and MAXSPACE-RDV. We present a $1/9$-approximation algorithm for MAXSPACE-R and a PTAS for MAXSPACE-RDV for the case that the number of slots is bounded by a constant. These are the first approximation algorithm and approximation schemes to these MAXSPACE variants. 

A PTAS is the best approximation ratio for MAXSPACE-RDV one can expect since it does not admit an FPTAS even for~${K = 2}$~\citep{kellerer2004introduction}. This variant is a generalization of the Multiple Knapsack Problem~\citep{chekuri2005polynomial}.

In future works, we will also consider MAXSPACE-RDV with the number of slots given in the instance, for which the ideas used in this work are not sufficient.

\textbf{Funding} This project was supported by S\~ao Paulo Research Foundation~(FAPESP) grants \mbox{\#2015/11937-9}, \mbox{\#2016/23552-7}, \mbox{\#2017/21297-2}, and \mbox{\#2020/13162-2}, and National Council for Scientific and Technological Development~(CNPq) grants \mbox{\#425340/2016-3}, \mbox{\#312186/2020-7}, and \mbox{\#311039/2020-0}.


\bibliographystyle{plainnat}
\bibliography{main}


\end{document}